\documentclass[conference]{IEEEtran}
\usepackage[utf8]{inputenc}
\usepackage{flushend}
\usepackage{cite}
\usepackage{amssymb,amsmath,amsthm} 
\usepackage{array}
\usepackage{fixltx2e}
\usepackage{url}
\usepackage{tikz}
\usetikzlibrary{dsp}

\theoremstyle{definition}
\newtheorem{definition}{Definition}
\newtheorem{remark}{Remark}
\theoremstyle{plain}
\newtheorem{proposition}{Proposition}

\newtheorem{lemma}{Lemma}

\def\R{{\ensuremath{\mathbb{R}}}}
\DeclareMathOperator{\E}{\mathbb{E}}
\let\P\undefined
\DeclareMathOperator{\P}{\mathbb{P}}
\let\leqa\lesssim 
\let\geqa\gtrsim 
\makeatletter
\def\@IEEEproof[#1]{\par\noindent{\itshape #1: }}
\makeatother

\begin{document}
\IEEEoverridecommandlockouts

\title{{\it At Every Corner}\/:\\\LARGE  Determining Corner Points of Two-User Gaussian Interference Channels}
\author{\IEEEauthorblockN{Olivier Rioul}
\IEEEauthorblockA{LTCI, T\'el\'ecom ParisTech,\\
Universit\'e Paris-Saclay, 75013 Paris, France\\
Email: olivier.rioul@telecom-paristech.fr}
}

\maketitle

\begin{abstract}
The corner points of the capacity region of the two-user Gaussian interference channel under strong or weak interference are determined using the notions of almost Gaussian random vectors, almost lossless addition of random vectors,  and almost linearly dependent random vectors. 
In particular, the ``missing'' corner point problem is solved in a manner that differs from previous works in that it avoids the use of integration over a continuum of SNR values or of Monge-Kantorovitch transportation problems.
\end{abstract}

\section{Introduction}

This work is about the complete determination of corner points of the capacity region of the two-user Gaussian interference channel.
Some classical ingredients are Fano's inequality, the data processing inequality (DPI), the maximum entropy (MaxEnt) property under a power constraint, the entropy power inequality (EPI), and the concavity of the entropy power.
Interestingly, only  weak forms of the latter two are required.
To these ingredients we add the notions of almost Gaussian random vectors, almost lossless addition of random vectors,  and almost linearly dependent random vectors. 

The determination of the second corner point under weak interference is the content of Costa's corner point conjecture~\cite{Costa85}.
This conjecture has been settled recently and independently by Polyanskiy and Wu~\cite{Polyanskiy15} (using optimal transport theory) and Bustin \textit{et al.}~\cite{Bustin14,Bustin15} (using the I-MMSE relation). The approach described here is  a natural continuation from previous works~\cite{CostaRioul14,CostaRioul15,RioulCosta15,RioulCosta16} that is very close in spirit to the solution of Polyanskiy and Wu. However, it is more direct because is sidesteps the notion of Wasserstein distance associated to a Monge–Kantorovich problem.

\section{Definitions and Notations}

Throughout the paper we consider zero-mean random vectors taking values in $\R^n$ and
let $\|\cdot\|$ denote the Euclidean norm in~$\R^n$.
Consider the two-user Gaussian interference channel in standard form (Fig.~\ref{GIC}):
\begin{equation}
\begin{split}
Y_1 &= X_1 + \sqrt{b} X_2 + Z_1 \\
Y_2 &= \sqrt{a} X_1 + X_2 + Z_2, 
 \end{split}
\end{equation}
where 
the joint distribution of the Gaussian noises $(Z_1,Z_2)$ at the decoder sides is not relevant as there is no cooperation between the receivers. We find it notationally convenient to set $Z_1=Z_2=Z$. The corresponding noise powers are $N_1=N_2=N$.
Sender $i=1,2$ produces a uniformly distributed $M_i$-ary message $W_i$, where $W_1$ and $W_2$ are independent.
Encoder~$i$ maps $W_i$ to a random vector $X_i\in\R^n$ of dimension~$n$ which satisfies the power constraint $\|X_i\|^2 \leq nP_i$.
Decoder~$i$ maps the output $Y_i$ to an $M_i$-ary decoded message~$\hat{W}_i$.

\begin{figure}[t]
\centering
\begin{tikzpicture}[xscale=1.3,yscale=1.1,>=latex]
\node(W1)at(0.2,0){$W_1$};
\node(X1)at(1,0){$X_1$};
\node[coordinate](X1')at(1.5,0){};
\node[dspadder](add1)at(4,0){};
\node(Y1)at(5,0){$Y_1$};
\node(W^1)at(5.8,0){$\widehat{W}_1$};

\node(W2)at(0.2,-2){$W_2$};
\node(X2)at(1,-2){$X_2$};
\node[coordinate](X2')at(1.5,-2){};
\node[dspadder](add2)at(4,-2){};
\node(Y2)at(5,-2){$Y_2$};
\node(W^2)at(5.8,-2){$\widehat{W}_2$};

\node(Z)at(4,-1){$Z$};

\draw[->] (X1) 
--(add1) node[midway,above]{$_1$};
\draw[->](add1)--(Y1);
\draw[->] (X2) 
--(add2) node[midway,below]{$_1$};
\draw[->](add2)--(Y2);

\node(P1)at(1,-0.5){($P_1$)};
\node(P2)at(1,-1.5){($P_2$)};

\draw[->](X1')--(add2) node[pos=0.7,below]{$\sqrt{a}$};
\draw[->](X2')--(add1) node[pos=0.7,above]{$\sqrt{b}$};

\draw[->] (Z)--(add1) node[midway,right]{($N_1$)};
\draw[->] (Z)--(add2) node[midway,right]{($N_2$)};
\end{tikzpicture}
\caption{Gaussian interference channel.}\label{GIC}
\end{figure}
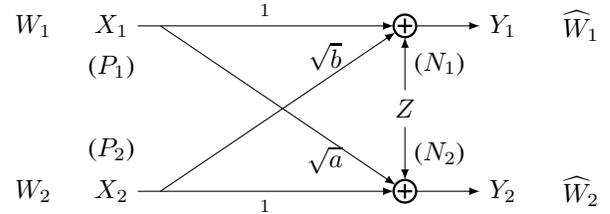

The {capacity region} of the channel may be defined as the set of all limit points of all sequences $(R_1,R_2)$ for which the corresponding sequence of encoding and decoding functions with $M_i= e^{n R_i}$  are such that $\P\{\hat{W}_i\ne W_i\}$ ($i=1,2$) tend to $0$  as $n\to+\infty$. 
Note that $R_1$, $R_2$, $W_1$, $W_2$, $X_1$, $X_2$, $Y_1$, $Y_2$, $Z_1$, $Z_2$ all depend on the dimension $n$. However, $P_1$, $P_2$ and $N$ are constants, independent of $n$. Because $n$ is taken arbitrarily large, it is convenient to use the following notation.
\begin{definition}[Almost Inequalities $\leqa$ and $\geqa$]
 Let $\epsilon(n)$ denote \emph{any} positive function of $n$ which tends to $0$ as $n\to+\infty$ (thus we can write, for example, $\epsilon(n)+\epsilon(n)=\epsilon(n)$). Given real number sequences $A_n,B_n$, we write $A_n\leqa B_n$ ($A_n$ is \textit{almost less} than $B_n$) if
\begin{equation}
A_n\leq B_n + n\epsilon(n) \iff B_n \geq A_n - n\epsilon(n).
 \end{equation}
We also write $B_n\geqa A_n$ ($B_n$ is \emph{almost greater} than $A_n$).
\end{definition}

The capacity region is a subset of the rectangle $R_1\leq C_1$, $R_2\leq C_2$, where $C_i=(1/2)\log (1+P_i/N_i)$ with two \emph{corner points}  $(C_1,C'_2)$ and $(C'_1,C_2)$.
A typical shape is shown in Fig.~\ref{corners}.
That $(C_1,C'_2)$ is a corner point is established by showing that it is achievable and that for any $(R_1,R_2)$ for which the associated probability of error tends to $0$ as $n\to+\infty$,
\begin{subequations}\label{corner}
\begin{equation}
nR_1\geqa nC_1 \implies nR_2 \leqa nC'_2. \label{corner1}
\end{equation}
That $(C'_1,C_2)$ is a corner point is similarly characterized by:
\begin{equation}
nR_2\geqa nC_2  \implies nR_1 \leqa nC'_1. \label{corner2}
\end{equation}
\end{subequations}
Achievability is generally not a problem and is done using classical ingredients such as random coding, onion peeling and rate splitting. 
Therefore, in this paper, we focus exclusively on the derivation of the converse~\eqref{corner}.

\begin{figure}[t]
\centering
\begin{tikzpicture}[scale=0.85]
\draw[->] (-0.2,0) -- (4.5, 0) node[right] {$R_1$};
\draw[->] (0,-0.2) -- (0,3.5) node[above left] {\smash[t]{$R_2$}};
\draw (4,0) node[below] {$C_1$} -- (4,1) node[dspnodeopen] {};
\draw (0,3) node[left] {$C_2$} -- (2,3) node[dspnodeopen] {};
\draw[very thin,loosely dashed] (0,1) node[left] {$C'_2$} --(4,1);
\draw[very thin,loosely dashed] (2,0) node[below] {$C'_1$} --(2,3);
\draw plot [domain=2:4] (\x,{5-\x+0.4*sin(pi*(\x/2-1) r)} );
\end{tikzpicture}
\caption{Corner points $(C_1,C'_2)$ and $(C'_1,C_2)$ of the capacity region (marked with circles).}\label{corners}
\end{figure}
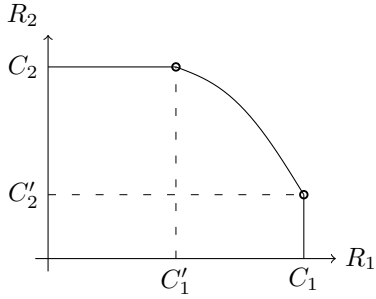

\section{Preliminaries}

Throughout the paper $X^G$ denotes a white Gaussian vector of the same variance as $X$.
\begin{lemma}\label{maxrate}
The condition $nR_1\geqa nC_1$ in~\eqref{corner1} implies
\begin{enumerate}
\item[(a)] $h(X_1+Z) \geqa h(X^G_1+Z)$;
\item[(b)] $I(X_1;Y_1)\geqa I(X_1;X_1+Z)$.
\end{enumerate}
\end{lemma}
\noindent The symmetrical lemma holds for~\eqref{corner2}.

\begin{IEEEproof}
By the classical derivation of the converse:
\begin{subequations}\label{directconverse}
\begin{align}
nR_1 = H(W_1)
&\leqa I(W_1;Y_1) &\text{(Fano)}    \label{HW}\\
&\leq I(X_1;Y_1) &\text{(DPI)}\label{DPI}\\
&\leq I(X_1;X_1+Z)&\text{\llap{(DPI again)}} \label{lossbyinterference}\\
&= h(X_1+Z)-h(Z)\\
&\leq nC_1 &\text{\llap{(MaxEnt)}}\label{maxent}
\end{align}
\end{subequations}
Thus $nR_1\geqa nC_1$ amounts to saying that all quantities in~\eqref{directconverse} are at distance $\leq n\epsilon(n)$. This implies, in particular, (a) from~\eqref{maxent} and (b) from~\eqref{lossbyinterference}. 
\end{IEEEproof}

\begin{remark}
Condition $nR_1\geqa nC_1$ also implies $I(W_1;Y_1)\geqa I(X_1;Y_1)$ which holds (with equality) if the encoder mapping is invertible. In that case $nR_1\geqa nC_1\iff $(a),(b).
\end{remark}

Lemma~\ref{maxrate} naturally leads to the following definitions.
\begin{definition}[AG and AL properties]\ 
Let $X$ have power constraint $\tfrac{1}{n\vphantom{_1}}\E\{\|X\|^2\}\leq P$. We say that $X$ is \emph{almost (white) Gaussian} (AG) if
\begin{equation}
h(X)\geqa  h(X^G).
\end{equation}
Let $Z$ and $Z'$ be mutually independent (not necessarily Gaussian) vectors, independent of $X$. We say that
$X+Z+Z'$ is \emph{almost lossless} (AL) compared to $X+Z$ (with respect to $X$) if
\begin{equation}
I(X;X+Z+Z') \geqa I(X;X+Z).
\end{equation}
\end{definition}
Thus (a), (b) in Lemma~\ref{maxrate} are equivalent to:
{\it\begin{enumerate}
\item[(a)] $X_1+Z$ is AG;
\item[(b)] $X_1\!+\!\sqrt{b}X_2+Z$ is  AL compared to $X_1+Z$ w.r.t. $X_1$.
\end{enumerate}
}\noindent
The latter condition means that adding interference $bX_2$ in $Y_1$ almost does not decrease information. This becomes vacuous in the case of no interference ($b=0$). If $b\ne 0$, condition (b) is equivalent to:
{\it
\begin{enumerate}
\item[(b$'$)] \!$X_1\!+\!\sqrt{b}X_2+Z$ is AL compared to $\sqrt{b}X_2+Z$ w.r.t. $X_2$.
\end{enumerate}
}\noindent
This is a direct consequence of the following lemma, which is particularly important as it allows one to pass from one transmission to the other (Fig.~\ref{forkgraph}).

\begin{lemma}[Fork Lemma]\label{fork}
Let $X_1$, $X_2$ and $Z$ be independent. If $X_1+X_2+Z$ is AL compared to $X_1+Z$ w.r.t.~$X_1$, then it is also AL compared to $X_2+Z$ w.r.t.~$X_2$.
\end{lemma}

\begin{IEEEproof}$I(X_2;X_1+X_2+Z)-I(X_2;X_2+Z) = h(X_1+X_2+Z)-h(X_1+Z)-h(X_2+Z)+h(Z)= I(X_1;X_1+X_2+Z) - I(X_1;X_1+Z)$.
\end{IEEEproof}

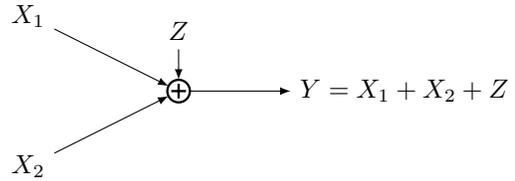
\begin{figure}[t]
\centering
\begin{tikzpicture}[>=latex]
\node(X1)at(0,1){$X_1$};
\node(X2)at(0,-1){$X_2$};
\node(sum)at(5,0){$Y=X_1+X_2+Z$};
\node[dspadder](add)at(2,0){};
\node(Z)at(2,0.8){$Z$};
\draw[->] (X1) --(add);
\draw[->] (X2) --(add);
\draw[->] (Z) --(add);
\draw[->] (add)--(sum);
\end{tikzpicture}
\caption{An illustration of the Fork Lemma.}\label{forkgraph}
\end{figure}

\bigskip

To simplify the derivations in the remainder of the paper, we restrict ourselves the case of a Gaussian {\sf Z}-interference channel with one of the interference parameters (e.g., $b$) equal to zero (Fig.~\ref{GZIC}):
\begin{equation}
\begin{split}
Y_1 &= X_1 + Z \\
Y_2 &= X_2+ \sqrt{a} X_1 + Z.
 \end{split}
\end{equation}
The general determination of corner points will follow in the general case of two-sided interference by noting that removing an interference link can only enlarge the capacity region, as explained in~\cite[Table I]{Costa85}.

\begin{figure}[h]
\centering
\begin{tikzpicture}[xscale=1.3,yscale=1.05,>=latex]
\node(W1)at(0.2,0){$W_1$};
\node(X1)at(1,0){$X_1$};
\node[coordinate](X1')at(1.5,0){};
\node[dspadder](add1)at(4,0){};
\node(Y1)at(5,0){$Y_1$};
\node(W^1)at(5.8,0){$\widehat{W}_1$};

\node(W2)at(0.2,-2){$W_2$};
\node(X2)at(1,-2){$X_2$};
\node[coordinate](X2')at(1.5,-2){};
\node[dspadder](add2)at(4,-2){};
\node(Y2)at(5,-2){$Y_2$};
\node(W^2)at(5.8,-2){$\widehat{W}_2$};

\node(Z)at(4,-1){$Z$};

\draw[->] (X1) 
--(add1);
\draw[->](add1)--(Y1);
\draw[->] (X2) 
--(add2); 
\draw[->](add2)--(Y2);

\node(P1)at(1,-0.5){($P_1$)};
\node(P2)at(1,-1.5){($P_2$)};

\draw[->](X1')--(add2) node[midway,below]{$\sqrt{a}$};

\draw[->] (Z)--(add1) node[midway,right]{($N_1$)};
\draw[->] (Z)--(add2) node[midway,right]{($N_2$)};
\end{tikzpicture}
\caption{Gaussian {\sf Z}-interference channel.}\label{GZIC}
\end{figure}
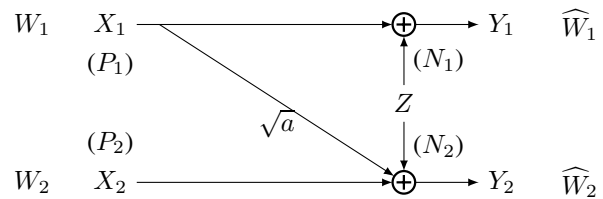

\section{Corner Points Under Strong interference}

The \emph{very strong interference} case ($a \geq 1+P_2/N$) is well-known~\cite{Carleial75}. One has $(C'_1=C_1,C'_2=C_2)$ and in this case there is no need to prove~\eqref{corner}.
For strong interference ($1\leq a\leq 1+P_2/N$) the corner points are known and given by~\eqref{strongcorners} below.
The usual derivation follows from that of the capacity region of the multiple access channel and from the result of Han and Kobayashi~\cite{HanKobayashi81} and Sato~\cite{Sato81}, who showed that both receivers should be able to decode both messages $W_1$ and $W_2$.
We offer a simple proof based on the following lemma.
\begin{lemma}\label{obvious}
Let $X_t\!=\!\sqrt{t}X$ and $Z$ be Gaussian independent of $X$. 
Then $I(X;X_t+Z)$, or $h(X_t+Z)$, is nondecreasing in~$t$.
\end{lemma}

\begin{IEEEproof}
Let $u=\frac{1}{t}$, $Z_u\!=\!\sqrt{u}Z$ so that $I(X;X_t+Z)=I(X;X+Z_u)$ and let $Z'$ be an independent copy of $Z$. By the DPI and the divisibility property of the Gaussian, $\forall\delta>0$, $I(X;X+Z_u) \geq I(X;X+Z_u+Z'_\delta) = I(X;X+Z_{u+\delta})$.
\end{IEEEproof}

\begin{proposition}\label{propstrong}
For the strong {\sf Z}-interference Gaussian channel,
\begin{subequations}\label{strongcorners}
\begin{align}
C'_1
&=  \frac{1}{2}\log\Bigl(1+\frac{aP_1+P_2}{N}\Bigr)  - C_2\notag\\
&= \frac{1}{2}\log\Bigl(1+\frac{a P_1}{P_2+N}\Bigr) \\
C'_2
&= \frac{1}{2}\log\Bigl(1+\frac{aP_1+P_2}{N}\Bigr)  - C_1 \notag\\
&= \frac{1}{2}\log\Bigl(1+\frac{(a-1)P_1+P_2}{P_1+N}\Bigr).
\end{align}
\end{subequations}
\end{proposition}

\begin{IEEEproof}[Proof of Proposition~\ref{propstrong}]
First suppose that $nR_1\geqa nC_1$. From Lemma~\ref{maxrate}, $X_1+Z$ is AG. Therefore, from~\eqref{HW}--\eqref{DPI} where index~1 is replaced by~2,
\begin{subequations}
\begin{align}
nR_2 &\leqa I(X_2;Y_2) \\
&= h(Y_2) -h(\sqrt{a}X_1\!+\!Z)\\
&\leq h(Y_2) -h(X_1+Z) &\text{\llap{(Lemma~\ref{obvious})}}\\
&\leqa h(Y_2)  - h(Z) - nC_1 & \text{(AG)}\\
&\leq nC'_2 &\text{(MaxEnt)}
\end{align}
\end{subequations}
which proves that $nR_2\leqa nC'_2$ (cf.~\eqref{corner1}).

Next suppose that $nR_2\geqa nC_2$. From Lemma~\ref{maxrate} written for transmission $2$, $X_2+Z$ is AG and $aX_1+X_2+Z$ is AL compared to $X_2+Z$ w.r.t. $X_2$. Since $a\ne 0$, by Lemma~\ref{fork}, $aX_1+X_2+Z$ is AL  compared to $aX_1+Z$ w.r.t.~$X_1$.
Therefore, from~\eqref{HW}--\eqref{DPI},
\begin{subequations}
\begin{align}
nR_1 &\leqa I(X_1;Y_1)
=I(X_1;X_1+Z)\\
&\leq I(X_1;\sqrt{a}X_1+Z) &\text{\llap{(Lemma~\ref{obvious})}}\\
&\leqa I(X_1;\sqrt{a}X_1+X_2+Z) & \text{(AL)}\\
&=h(Y_2) -h(X_2+Z) \\
&\leqa h(Y_2) -h(Z) -nC_2 &\text{(AG)}\\
&\leq nC'_1 &\text{(MaxEnt)}
\end{align}
\end{subequations}
which proves that $nR_1\leq nC'_1$ (cf.~\eqref{corner2}).
\end{IEEEproof}

\section{Sato's Corner  Point}

For weak interference $a< 1$, Sato~\cite{Sato78} (see also~\cite{Sason04}) has found that the first corner point is given by~\eqref{satocorner} below.
The usual derivation follows from the equivalence between Gaussian {\sf Z}-interference channel and a ``fully'' degraded version proved in~\cite{Costa85}, the fact that it can be considered as a broadcast channel with input power given by $P_1+P_2$~\cite{Sato78}, and the derivation of the capacity region of the Gaussian (degraded) broadcast channel by Bergmans~\cite{Bergmans74}. 
We give a simple proof based on the following lemma which is a direct consequence of the EPI.

\begin{lemma}\label{AWGZ}
Let $X_t\!=\!\sqrt{t}X$ and $Z$ be Gaussian independent of $X$. If $X+Z$ is AG then so is $X_t+Z$ for any $0<t< 1$.
\end{lemma}
\begin{IEEEproof}
Let $u=1/t>1$, $Z_u=\sqrt{u}Z$ and let $Z'$ be an independent copy of $Z$. By the DPI for divergence and the divisibility property of the Gaussian,
$h(X^G_t +Z)-h(X_t+Z) = h(X^G + Z_u)-h(X+Z_u) = h(X^G+Z+Z'_{u-1})-h(X+Z+Z'_{u-1})=D(X+Z+Z'_{u-1}\|X^G+Z+Z'_{u-1})\leq D(X+Z\|X^G+Z)=h(X^G+Z)-h(X+Z)$.
\end{IEEEproof}

\begin{remark}
By noting that $X$ is AG if and only if its entropy power $N(X)$ satisfies $N(X) \geq N(X^G) -\epsilon(n)$, it is readily seen 
that the general EPI $N(X+Y)\geq N(X)+N(Y)$ for independent $X,Y$ implies that if $X$ and $Y$ are AG, then so is $X+Y$~\cite{RioulCosta15}.  Thus the conclusion of Lemma~\ref{AWGZ} is also obtained using the EPI where one of the variables is Gaussian: $N(X+Z)\geq N(X)+N(Z)$.  

It is interesting to note, however, that the EPI is not even required: only the DPI applied to divergence was necessary in the above proof, which is strictly weaker than the EPI. In fact, $D(X+Z\|X^G+Z)\leq D(X\|X^G)$ is equivalent to $N(X+Z) \geq N(X) + N(Z) \cdot \bigl(N(X)/N(X^G)\bigr)$ where $N(X)/N(X^G)\leq 1$.
\end{remark}

\begin{proposition}
For the  weak {\sf Z}-interference  Gaussian channel,
\begin{equation}\label{satocorner}
C'_2=\frac{1}{2}\log \Bigl(1+ \frac{P_2}{aP_1+N}\Bigr).
\end{equation}
\end{proposition}

\begin{IEEEproof}
Suppose that $nR_1\geqa nC_1$. From Proposition~\ref{maxrate}, $X_1+Z$ is AG. By Lemma~\ref{AWGZ}, $\sqrt{a}X_1+Z$ is also AG. Therefore, from~\eqref{HW}--\eqref{DPI} written for $i=2$,
\begin{subequations}
\begin{align}
nR_2 &\leqa I(X_2;Y_2)
= h(Y_2) -h(\sqrt{a}X_1+Z)\\
&\leqa h(Y_2) - h(\sqrt{a}X_1^G+Z) &\text{(AG)}\\
&\leq nC'_2 &\text{\llap{(MaxEnt)}}
\end{align}
\end{subequations}
which proves that $nR_2\leqa nC'_2$ (cf.~\eqref{corner1}).
\end{IEEEproof}

\section{Almost Linear Dependence}

For any two (zero-mean) $n$-dimensional random vectors $U,V$ with finite average powers  we define their \emph{correlation coefficient} by
\begin{equation}
\rho(U,V) = \frac{\E \{ U\cdot V \} }{ \sqrt{\E{\|U\|^2\}\E{\|V\|^2\}}}}}
\end{equation}
where $\cdot$ denotes the scalar product. By Cauchy-Schwarz inequality\footnote{%
This particular instance of Cauchy-Schwarz inequality can be proved by considering the discriminant of the nonnegative quadratic form $\lambda\mapsto \E\{\|U+\lambda V\|^2\}$. Alternatively, one has $|\E \{ U\cdot V \}| \leq \sum_{i=1}^n |\E\{U_iV_i\}| \leq 
\sum_{i=1}^n \sqrt{\E{|U_i|^2\}}}\sqrt{\E{|V_i|^2\}}}\leq \sqrt{\E{\|U\|^2\}\E{\|V\|^2\}}}}$
where the Cauchy-Schwarz inequality is applied twice (for random variables and for vectors).
}
one has $|\rho(U,V)|\leq 1$ with equality if and only if $U$ and $V$ are \emph{linearly dependent} in the sense that $U=\lambda V$ a.e. for some $\lambda\in\R$.

\begin{definition}[ALD property]
We say that $U$ and $V$ are \emph{almost linearly dependent} (ALD) if
\begin{equation}
1-|\rho(U,V)| \leq \epsilon(n). 
\end{equation}
(Recall that $\epsilon(n)$ denotes {any} positive function of $n$ which tends to $0$ as $n\to+\infty$.)
\end{definition}

We now consider $Y=X_2+Z$ of variance $Q\leq P_2+N$ and the interference term $X=\sqrt{a} X_1$. 

\begin{remark}\label{contdens}
Since $Z$ is Gaussian, it is proven in~\cite[App.~II.A]{GengNair14} that $Y=X_2+Z$ has a \emph{continuous} density (see also~\cite[Lemma~1]{Rioul11}\footnote{%
In fact, the density of $Y$ is indefinitely differentiable, bounded, positive, tends to zero at infinity and all its derivatives are also bounded and tend to zero at infinity~\cite[App. B]{Rioul16}; but we shall not need this result here.
}).  Similarly $X+Y=(\sqrt{a} X_1+X_2)+Z$ also has a continuous density. In contrast, $X$ is proportional to a code distribution  that is typically discrete.
\end{remark}

Clearly $Y^G=X_2^G+Z$ satisfies the inequality $h(Y) \leq h(Y^G)$. However, the interference term $X$ might very well be such that the opposite inequality $h(X+Y) \geq h(X+Y^G)$ holds after addition. We now aim at bounding the difference $h(X+Y) - h(X+Y^G)$.

\begin{lemma}\label{truc} One has
\begin{equation}
 h(X+Y) - h(X+Y^G) \leq c \cdot n \cdot \sqrt{1-\rho(Y,Y^G)}
\end{equation}
where $c$ is a constant (independent of $n$).
\end{lemma}

\begin{proof} 
The continuous p.d.f. $q$ of $X+Y^G$  takes the form
\begin{equation}
q(u) = \E \{q(u|X)\} 
=  \frac{  \E \exp \Bigl(- \dfrac{\|u-X\|^2}{2Q}\Bigr) }{(2\pi)^{n/2}Q^n}. 
\end{equation}
Since $D(X+Y\|X+Y^G)\geq 0$, we have
\begin{equation}\label{schmurf}
  h(X+Y) - h(X+Y^G)
  \leq \E \log \frac{q(X+Y^G)}{q(X+Y)}.
\end{equation}
where
\begin{equation}
 \log \frac{q(\tilde{u})}{q(u)}= \log \frac{  \E \exp \Bigl(- \dfrac{\|\tilde{u}-X\|^2}{2Q}\Bigr) }{  \E \exp \Bigl(- \dfrac{\|{u}-X\|^2}{2Q}\Bigr) }.
\end{equation}
Now for any $u\in\R^n$,
$
 \|u-X\|^2 - \|\tilde{u}-X\|^2 = \|u\|^2 - \|\tilde{u}\|^2 +2 X \cdot (\tilde{u}-u)
 \leq \|u\|^2 - \|\tilde{u}\|^2 +2 \sqrt{a n P_1} \|\tilde{u}-u\|$. 
It follows that
\begin{equation}
  \log \frac{q(\tilde{u})}{q(u)} \leq   \frac{\|u\|^2 - \|\tilde{u}\|^2 +2 \sqrt{a n P_1} \|\tilde{u}-u\|}{2Q}
\end{equation}
where the identical terms $\E \exp (- {\|\tilde{u}-X\|^2}/{2Q})$ in the numerator and denominator were cancelled. Plugging this inequality into~\eqref{schmurf} and noting that $X+Y^G-(X+Y)=Y^G-Y$ we obtain
\begin{align}
  h(X+Y) &- h(X+Y^G) \notag\\&\leq  {\E\{ \|X+Y\|^2\}}/{2Q} - {\E\{ \|X+Y^G\|^2\}}/{2Q} \notag \\ + \frac{\sqrt{a n P_1}}{Q} &\sqrt{\E \{\|Y\|^2\} + \E\{\|Y^G\|^2\} - 2 \E\{Y \cdot Y^G\}\}} \label{schglorps}\\
  &=  n \sqrt{{2aP_1/Q}} \cdot \sqrt{1-\rho(Y,Y^G)}
\end{align}
where the first two terms in~\eqref{schglorps} were cancelled.
\end{proof}

The result of Lemma~\ref{truc} shows that if $Y$ and $Y^G$ are ALD such that $1-\rho(Y,Y^G)\leq \epsilon(n)$, then 
$h(X+Y) - h(X+Y^G)\leqa 0$. In other words $h(X+Y^G) - h(X+Y)\geqa 0$ is almost positive: it can be negative, but not by much.
In order to obtain a value $\rho(Y,Y^G)$ close to one, the next lemma shows that is sufficient to assume a dependence of the form $Y=F(Y^G)$ where $F$ is ``almost linear''.

\begin{lemma}\label{NG2G}
One can always assume that $Y=F(Y^G)$ where the change of variable $F$ has a triangular Jacobian matrix $\mathbf{J}$ with positive diagonal elements such that 
\begin{equation}\label{traceeq}
\rho(Y,Y^G) = \frac{1}{n}\E\{ \mathrm{Tr} (\mathbf{J}) \}\geq 0.
\end{equation}
\end{lemma}
\noindent Of course, a truly linear dependence of the form $Y=\lambda Y^G$ implies $\lambda=1$ (since $Y$ and $Y^G$ have the same variance), hence $\mathbf{J}=\mathbf{I}$ (identity matrix), in keeping with the fact that $\rho(Y,Y^G)=1$ in this case.

\begin{proof}
The change of variable of this lemma is well known as \emph{Kn\"othe map} in the theory of convex bodies~\cite[p.~126]{MilmanSchechtman86},\cite[p.~312]{Schneider93},~\cite[Thm.~3.4]{GiannopoulosMilman04}, \cite[Thm.~1.3.1]{ArtsteinGiannopoulosMilman15}. For completeness we give Kn\"othe's proof~\cite{Knothe57}. 
By Remark~\ref{contdens},  $Y$ has a continuous density. 
For each $y^G_1\in\R$, define $F_1(y^G_1)$ such that 
\begin{equation}
\int_{-\infty}^{F_1(y^G_1)} p_{Y_1} = \int_{-\infty}^{y^G_1} p_{Y^G_1}.
 \end{equation}
Clearly $F_1$ is increasing and differentiating gives
\begin{equation}
p_{Y_1}(F_1(y^G_1)) \;  \frac{\partial F_1}{\partial y^G_1}(y^G_1) = p_{Y^G_1}(y^G_1)
\end{equation}
which proves the result in one dimension: $Y_1$ has the same distribution as $F_1(Y^G_1)$ where $\dfrac{\partial F_1}{\partial y^G_1}$ is positive.
Next for each $y^G_1,y^G_2$ in $\R$, define  $F_2(y^G_1,y^G_2)$ such that 
\begin{equation}
\int_{-\infty}^{F_2(y^G_1,y^G_2)}\!\!\! p_{Y_1,Y_2}(F_1(y^G_1),\,\cdot\,) \;  \frac{\partial F_1}{\partial y^G_1}(y^G_1) = 
\int_{-\infty}^{y^G_2}\!\! p_{Y^G_1,Y^G_2}(y^G_1,\,\cdot\,)
\end{equation} 
Again $F_2$ is increasing in $y^G_2$ and differentiating gives
\begin{multline}
p_{Y_1,Y_2}(F_1(y^G_1),F_2(y^G_1,y^G_2)) \;  \frac{\partial F_1}{\partial y^G_1}(y^G_1)  \frac{\partial F_2}{\partial y^G_2}(y^G_1,y^G_2) \\=  p_{Y^G_1,Y^G_2}(y^G_1,y^G_2).
\end{multline} 
Continuing in this manner we arrive at
\begin{multline}
\begin{split}
&p_{Y_1,Y_2,\ldots,Y_n}(F_1(y^G_1),F_2(y^G_1,y^G_2),\ldots,F_n(y^G_1,y^G_2,\ldots,y^G_n))
 \\&\qquad\times   \frac{\partial F_1}{\partial y^G_1}(y^G_1)  \frac{\partial F_2}{\partial y^G_2}(y^G_1,y^G_2) \cdots \frac{\partial F_n}{\partial y^G_n}(y^G_1,y^G_2,\ldots,y^G_n) 
 \end{split}
\\=  p_{Y^G_1,Y^G_2,\ldots,Y^G_n}(y^G_1,y^G_2,\ldots,y^G_n) 
\end{multline}
which shows that $Y$ has the same distribution as $F(Y^G)=\bigl(F_1(Y^G_1),F_2(Y^G_1,Y^G_2),\ldots,
F_n(Y^G_1,Y^G_2,\ldots,Y^G_n)\bigr)$. The Jacobian matrix $\mathbf{J}$ of $F$ is triangular with positive diagonal elements are positive since by construction each $F_k$ is increasing in $y^G_k$. For convenience we choose to define $(Y,Y^G)$ such that $Y=F(Y^G)$. 
By Stein's lemma,
\begin{align}
\rho(Y,Y^G) &= \frac{1}{nQ}\sum_{i=1}^n \E(Y^G_i \cdot F_i(Y^G))\\
&=\frac{1}{n} \sum_{i=1}^n \E(\frac{\partial F_i}{\partial y_i^G}(Y^G))= \frac{1}{n}\E\{\mathrm{Tr} (\mathbf{J})\}.\qedhere
\end{align}
\end{proof}

\begin{proposition}\label{therealthing}
If $Y$ is AG, then  $Y$ and $Y^G$ are ALD and 
\begin{equation}
I(X;X+Y^G)\geqa I(X;X+Y).
\end{equation}
\end{proposition}
\noindent The latter equation also reads, with our previous notations,
\begin{equation}
I(X_1;\sqrt{a}X_1 + X_2^G +Z)
\geqa I(X_1;\sqrt{a}X_1 + X_2 +Z).
 \end{equation}

\begin{proof}
By making the change of variable in the expression of $Y=F(Y^G)$ one obtains
\begin{equation}
h(Y)=h(F(Y^G))=h(Y^G)+\E \log \det \mathbf{J}
\end{equation}
Thus, since $Y$ is AG,  $\E \log \det \mathbf{J}\geqa 0$.
On the other hand from~\eqref{traceeq} by Hadamard's inequality,
\begin{align}
\rho(Y,Y^G) = \frac{1}{n}\E\{\mathrm{Tr} (\mathbf{J}) \}&\geq \E\{\sqrt[n]{ \det \mathbf{J}}\}\\
&\geq e^{\frac{1}{n}\E\log\det\mathbf{J}}
\end{align}
which shows that $Y$ and $Y^G$ are ALD, such that $1-\rho(Y,Y^G) \leq \epsilon(n)$. From Lemma~\ref{truc} it follows that $h(X+Y)-h(X+Y^G)\leqa 0$, hence $I(X;X+Y) = h(X+Y)-h(Y) \leqa h(X+Y^G)-h(Y^G) = I(X;X+Y^G)$.
\end{proof}

\section{The ``Missing'' Corner  Point}

For weak interference $a<1$, Costa~\cite{Costa85} has stated that the second corner point is given by~\eqref{costacorner} below.
A problematic issue in the proof was detected by Sason~\cite{Sason04} 
\nocite{Sason15}
and the corner point has been later dubbed ``missing''~\cite{Kramer06}.
Recently, Polyanskiy and Wu~\cite{Polyanskiy15} solved the missing corner point problem using optimal transport theory by showing Lipschtiz continuity of differential entropy with respect to the Wasserstein distance and Talagrand's transportation-information inequality. An independent solution using the I-MMSE approach was given by Bustin \textit{et al.}~\cite{Bustin14,Bustin15} for a restricted subset of inputs---and later more generally---by integration of the MMSE over a continuum of SNR values.
We provide yet another solution to the problem in continuation of previous investigations~\cite{CostaRioul14,CostaRioul15,RioulCosta15,RioulCosta16} that is close to Polyanskiy and Wu's but sidesteps the use of the Wasserstein distance. Our proof is based on Prop.~\ref{therealthing} and the following lemma.

\begin{lemma}\label{conc}
Let $Z$ be Gaussian independent of $X$ and write $Z_u\!=\!\sqrt{u}Z$.
For any positive $u< u'< u''$, there exists $\mu$ constant independent of $n$ such that
\begin{equation}
\begin{split}
I(X;X+&Z_{u'})-I(X;X+Z_u) \\&\geq \mu \cdot \bigl(  I(X;X+Z_{u''})-I(X;X+Z_{u'})  \bigr) 
\end{split}
 \end{equation}
 Consequently, $ I(X;X+Z_{u''})\geqa I(X;X+Z_{u'}) $  implies $I(X;X+Z_{u'})\geqa I(X;X+Z_u)$.
\end{lemma}

\begin{IEEEproof}
Letting $t=1/u>t'=1/u'>t''=1/u''$,  it is equivalent to show that
$I(X;X_{t'}+Z)-I(X;X_t+Z) \geq \mu \cdot \bigl(  I(X;X_{t''}+Z)-I(X;X_{t'}+Z)  \bigr)$.
But this holds with $\mu=\frac{t-t'}{t'-t''}$ by concavity of $t\mapsto I(X;X_t+Z)$. 
\end{IEEEproof}

\begin{remark}
The concavity of $I(X;X_t+Z)$ or $h(X_t+Z)$ is a consequence of the concavity of the entropy power~\cite{Costa85a} $N(X_t+Z)$ but is strictly weaker as remarked in~\cite{Polyanskiy15}, since a concave function is not always exponentially concave. In fact it can be shown~\cite{Rioul11} that the concavity of $N(X_t+Z)$ is equivalent to the concavity of $N(X+Z_t)$. By taking the logarithm, this implies concavity of both $h(X_t+Z)$ and $h(X+Z_t)$. While the latter can be shown directly using the DPI~\cite{FahsAbouFaycal15}, the former requires de Bruijn's identity or the I-MMSE relation~\cite{GuoShamaiVerdu05}.
\end{remark}

\begin{proposition}
For the weak {\sf Z}-interference Gaussian channel,
\begin{equation}\label{costacorner}
C'_1=\frac{1}{2}\log \Bigl(1+ \frac{aP_1}{P_2+N}\Bigr).
\end{equation}
\end{proposition}

\begin{IEEEproof}
Suppose that $nR_2\geqa nC_2$. From Proposition~\ref{maxrate} written for transmission $2$, $X_2+Z$ is AG and adding interference $aX_1$ in $Y_2=\sqrt{a}X_1+X_2+Z$ is AL w.r.t. $X_2$. Since $a\ne 0$, by the Fork Lemma (Lemma~\ref{fork}), this implies that adding $X_2$ in $Y_2=\sqrt{a}X_1+X_2+Z$ is AL compared to $\sqrt{a}X_1+Z$ w.r.t.~$X_1$.

Therefore, 
\begin{subequations}
\begin{align}
\!\!\!nC'_1&= h(\sqrt{a}X_1^G+X_2^G+Z) - h(X^G_2+Z) \\
&\geq h(\sqrt{a}X_1+X_2^G+Z) - h(X^G_2+Z)  &\qquad\;\;\;\text{\llap{(MaxEnt)}}\\
&= I(X_1;\sqrt{a}X_1 + X_2^G +Z)\\
&\geqa I(X_1;\sqrt{a}X_1 + X_2 +Z)&\text{\llap{(Prop.~\ref{therealthing})}}\\
&\geqa I(X_1;\sqrt{a}X_1+Z) & \text{(AL)}\\
&\geqa I(X_1;\sqrt{a}X_1+\sqrt{a}Z) & \text{\llap{(Lemma~\ref{conc})}}\\
&=I(X_1;X_1+Z)=I(X_1;Y_1)\\
&\geqa nR_1  &\text{\llap{(see~\eqref{HW}--\eqref{DPI})}}
\end{align}
\end{subequations}
which proves that $nR_1\leqa nC'_1$ (cf.~\eqref{corner2}).
Notice that we have used  Lemma~\ref{conc} for $u=aN$, $u'=N$ and $u''=P_2+N$, in the form:  $I(X_1;\sqrt{a}X_1 + X_2^G +Z)\geqa I(X_1;\sqrt{a}X_1+Z)$ implies $I(X_1;\sqrt{a}X_1+Z)\geqa   I(X_1;\sqrt{a}X_1+\sqrt{a}Z)$.
\end{IEEEproof}

\section{Conclusion}

In this work, a complete determination of corner points of the capacity region of the two-user Gaussian interference channel is carried out, using the notions of almost Gaussian random vectors, almost lossless addition of random vectors,  and almost linearly dependent random vectors.  The resulting proofs use basic properties of Shannon's information theory.
Interestingly, only weak forms the entropy power inequality and the concavity of the entropy power are required.
This approach does not aim at finding best possible constants but yields a rigorous proof  for the determination of Costa's ``missing'' corner point which can be thought of as a variation of the solution of Polyanskiy and Wu which does not recourse to optimal transport theory nor to estimation theory.

\section*{Acknowledgments}
The author would like to thank Flavio Calmon,  Max Costa, Michèle Wigger and  Yihong Wu for their discussions.


\begin{thebibliography}{10}
\providecommand{\url}[1]{#1}
\csname url@samestyle\endcsname
\providecommand{\newblock}{\relax}
\providecommand{\bibinfo}[2]{#2}
\providecommand{\BIBentrySTDinterwordspacing}{\spaceskip=0pt\relax}
\providecommand{\BIBentryALTinterwordstretchfactor}{4}
\providecommand{\BIBentryALTinterwordspacing}{\spaceskip=\fontdimen2\font plus
\BIBentryALTinterwordstretchfactor\fontdimen3\font minus
  \fontdimen4\font\relax}
\providecommand{\BIBforeignlanguage}[2]{{%
\expandafter\ifx\csname l@#1\endcsname\relax
\typeout{** WARNING: IEEEtran.bst: No hyphenation pattern has been}%
\typeout{** loaded for the language `#1'. Using the pattern for}%
\typeout{** the default language instead.}%
\else
\language=\csname l@#1\endcsname
\fi
#2}}
\providecommand{\BIBdecl}{\relax}
\BIBdecl

\bibitem{Costa85}
M.~H.~M. Costa, ``On the {Gaussian} interference channel,'' \emph{{IEEE} Trans.
  Inf. Theory}, vol.~31, no.~5, pp. 607--615, Sept. 1985.

\bibitem{Polyanskiy15}
Y.~Polyanskiy and Y.~Wu, ``Wasserstein continuity of entropy and outer bounds
  for interference channels,'' \emph{{IEEE} Trans. Inf. Theory}, vol.~62,
  no.~7, pp. 3992--4002, July 2016.

\bibitem{Bustin14}
R.~Bustin, H.~V. Poor, and S.~Shamai, ``The effect of maximal rate codes on the
  interfering message rate,'' in \emph{Proc. ISIT'14}, Honolulu, Hawaii, USA,
  July 2014, pp. 91--95, longer draft available at
  http://arxiv.org/abs/1404.6690.

\bibitem{Bustin15}
------, ``Optimal point-to-point codes in interference channels: {An}
  incremental approach,'' 2015, draft at http://arxiv.org/abs/1510.08213.

\bibitem{CostaRioul14}
M.~H.~M. Costa and O.~Rioul, ``From almost {Gaussian} to {Gaussian},'' in
  \emph{AIP Proc. Int. Workshop on Bayesian Inference and Maximum Entropy
  Methods (MaxEnt)}, Amboise, France, Sept. 21--26, 2014.

\bibitem{CostaRioul15}
------, ``From almost {Gaussian} to {Gaussian}: {Bounding} differences of
  differential entropies,'' in \emph{Information Theory and Applications
  Workshop (ITA 2015)}, San Diego, Feb. 2--6 2015.

\bibitem{RioulCosta15}
O.~Rioul and M.~H.~M. Costa, ``Almost there: {Corner} points of {Gaussian}
  interference channels,'' in \emph{Information Theory and Applications
  Workshop (ITA 2015)}, San Diego, Feb. 2--6 2015.

\bibitem{RioulCosta16}
------, ``On some almost properties,'' in \emph{Information Theory and
  Applications Workshop (ITA 2016)}, San Diego, Jan.~31--Feb.~5 2016.

\bibitem{Carleial75}
A.~B. Carleial, ``A case where interference does not reduce capacity,''
  \emph{{IEEE} Trans. Inf. Theory}, vol.~21, no.~5, pp. 569--570, Sept. 1975.

\bibitem{HanKobayashi81}
T.~S. Han and K.~Kobayashi, ``A new achievable rate region for the interference
  channel,'' \emph{{IEEE} Trans. Inf. Theory}, vol.~27, no.~1, pp. 49--60, Jan.
  1981.

\bibitem{Sato81}
H.~Sato, ``The capacity of the {Gaussian} interference channel under strong
  interference,'' \emph{{IEEE} Trans. Inf. Theory}, vol.~27, no.~6, pp.
  786--788, Nov. 1981.

\bibitem{Sato78}
------, ``On degraded {Gaussian} two-user channels,'' \emph{{IEEE} Trans. Inf.
  Theory}, vol.~24, no.~5, pp. 637--640, Sept. 1978.

\bibitem{Sason04}
I.~Sason, ``On achievable rate regions for the {Gaussian} interference
  channel,'' \emph{{IEEE} Trans. Inf. Theory}, vol.~50, no.~6, pp. 1345--1356,
  June 2004.

\bibitem{Bergmans74}
P.~P. Bergmans, ``A simple converse for broadcast channels with additive white
  {Gaussian} noise,'' \emph{{IEEE} Trans. Inf. Theory}, vol.~20, no.~2, pp.
  279--280, March 1974.

\bibitem{GengNair14}
Y.~Geng and C.~Nair, ``The capacity region of the two-receiver {Gaussian}
  vector broadcast channel with private and common messages,'' \emph{{IEEE}
  Trans. Inf. Theory}, vol.~60, no.~4, pp. 2087--2104, Apr. 2014.

\bibitem{Rioul11}
O.~Rioul, ``Information theoretic proofs of entropy power inequalities,''
  \emph{{IEEE} Trans. Inf. Theory}, vol.~57, no.~1, pp. 33--55, Jan. 2011.

\bibitem{Rioul16}
------, ``Yet another proof of the entropy power inequality,'' \emph{{IEEE}
  Trans. Inf. Theory}, to appear, available at
  \url{http://arxiv.org/abs/1606.05969}.

\bibitem{MilmanSchechtman86}
V.~D. Milman and G.~Schechtman, \emph{Asymptotic Theory of Finite Dimensional
  Normed Spaces}, ser. Lecture Notes in Mathematics.\hskip 1em plus 0.5em minus
  0.4em\relax Springer, 1986, vol. 1200.

\bibitem{Schneider93}
R.~Schneider, \emph{Convex Bodies: The Brunn-Minkowski Theory}.\hskip 1em plus
  0.5em minus 0.4em\relax Cambridge University Press, 1993.

\bibitem{GiannopoulosMilman04}
A.~A. Giannopoulos and V.~D. Milman, ``Asymptotic convex geometry: {A} short
  overview,'' in \emph{Different Faces of Geometry}, S.~Donaldson,
  Y.~Eliashberg, and M.~Gromov, Eds.\hskip 1em plus 0.5em minus 0.4em\relax
  Springer, 2004, vol.~3, pp. 87--162.

\bibitem{ArtsteinGiannopoulosMilman15}
S.~Artstein-Avidan, A.~Giannopoulos, and V.~D. Milman, \emph{Asymptotic
  Geometric Analysis I}.\hskip 1em plus 0.5em minus 0.4em\relax Amer. Math.
  Soc., 2015.

\bibitem{Knothe57}
H.~Kn\"othe, ``Contributions to the theory of convex bodies,'' \emph{Michigan
  Math. J.}, vol.~4, pp. 39--52, 1957.

\bibitem{Sason15}
I.~Sason, ``On the corner points of the capacity region of a two-user
  {Gaussian} interference channel,'' \emph{{IEEE} Trans. Inf. Theory}, vol.~61,
  no.~7, pp. 3682--3697, July 2015.

\bibitem{Kramer06}
G.~Kramer, ``Review of rate regions for interference channels,'' in \emph{Proc.
  IEEE Int. Zurich Seminar on Communications (IZS)}, Feb. 22--24, 2006, pp.
  162--165.

\bibitem{Costa85a}
M.~H.~M. Costa, ``A new entropy power inequality,'' \emph{{IEEE} Trans. Inf.
  Theory}, vol.~31, no.~6, pp. 751--760, Nov. 1985.

\bibitem{FahsAbouFaycal15}
J.~Fahs and I.~Abou-Fay\c{c}al, ``A new tight upper bound on the entropy of
  sums,'' \emph{Entropy}, vol.~17, pp. 8312--8324, Dec. 2015.

\bibitem{GuoShamaiVerdu05}
D.~Guo, S.~S. (Shitz), and S.~Verd\'u, ``Mutual information and minimum
  mean-square error in gaussian channels,'' \emph{{IEEE} Trans. Inf. Theory},
  vol.~51, no.~4, pp. 1261--1282, April 2005.

\end{thebibliography}


\vspace*{0.01ex}
\end{document}